\documentclass[a4paper,11pt]{amsart}
\usepackage[centering,height=9in,width=6.5in,headsep=0.5in,foot=0.5in]{geometry}
\usepackage[pdftitle={Title},pdfauthor={Authors}]{hyperref}
\usepackage{lineno}
\usepackage{apalike}
\usepackage
{amsaddr}

\renewcommand{\emph}{\textbf}

\usepackage{amsmath}
\usepackage{amsfonts,amssymb}
\usepackage{amsthm}
\usepackage[PostScript=dvips]{diagrams}

\theoremstyle{plain}
\newtheorem{theorem}{Theorem}[section]

\newtheorem{proposition}[theorem]{Proposition}

\newtheorem{corollary}[theorem]{Corollary}
\newtheorem{question}[theorem]{Question}
\theoremstyle{definition}  

\usepackage{tikz}
\usetikzlibrary{decorations.pathreplacing}
\usetikzlibrary{snakes}
\usetikzlibrary{backgrounds}
\usetikzlibrary{arrows}
\usetikzlibrary{3d,fadings}
\tikzfading [name=radialfade, inner color=transparent!0, outer color=transparent!100]
\usetikzlibrary{shapes}
\usetikzlibrary{calc}

\usepackage{xcolor}

\usepackage{colonequals}
 
\newcommand{\setdef}  [2]{\left\{#1 \;\middle|\; #2\right\}}             

\newcommand{\adjointfdec}[4]{(#1 \vdash #2)\colon \xymatrix{#3 \ar@<.5ex>[r]^{f_*} & \ar@<.5ex>[l]^{f^*} #4}}

\newcommand{\FaceLattEM}[1]{\mathcal{F}_\Sigma}
\newcommand{\SuppLattEM}[1]{\mathcal{P}_\Sigma}

\newcommand{\relint}{\mathsf{relint}\,}
\newcommand{\carr}{\mathsf{carr}\,}
\newcommand{\supp}{\mathsf{supp}\,}

\def\poss{\mathsf{poss}\,}

\usepackage{nicefrac}

\newcommand{\FF}{\mathcal{F}}
\newcommand{\SP}{\mathcal{S}}
\newcommand{\Hpos}{\mathcal{H}_{\geq \mathbf{0}}}
\newcommand{\vv}{\mathbf{v}}
\newcommand{\zv}{\mathbf{0}}
\newcommand{\Aff}{\mathsf{Aff}}
\newcommand{\NS}{\mathcal{N}}
\newcommand{\HH}{\mathcal{H}}
\newcommand{\VV}{\mathcal{V}}
\newcommand{\ie}{\textit{i.e.}~}
\newcommand{\MM}{\mathcal{M}}
\newarrow{epi}----{>>}
\newcommand{\IFF}{\; \Longleftrightarrow \;}

\newcommand{\IMP}{\; \Rightarrow \;}
\newcommand{\xx}{\mathbf{x}}
\newcommand{\yy}{\mathbf{y}}
\newcommand{\uu}{\mathbf{u}}
\newcommand{\zz}{\mathbf{z}}
\newcommand{\av}{\mathbf{a}}
\newcommand{\Prob}{\mathsf{Prob}}
\newcommand{\rarr}{\rightarrow}
\newcommand{\card}[1]{| #1|}
\newcommand{\Real}{\mathbb{R}}
\newcommand{\Rn}{\Real^n}
\newcommand{\Fp}{\FF^{+}}
\newcommand{\SL}{\mathcal{S}}
\newcommand{\Sg}{\Sigma}
\newcommand{\Rgeq}{\Real_{\geq 0}}

\newcommand{\Bool}{\mathbb{B}}
\newcommand{\LL}{\mathcal{L}}
\newcommand{\MSC}{\mathsf{MSC}}
\newcommand{\LD}{\mathsf{LD}}
\newcommand{\Conv}{\mathsf{Conv}}
\newcommand{\sx}{\xx^{\sigma}}
\newcommand{\Fx}{F_{\xx}}
\newcommand{\Fy}{F_{\yy}}
\newcommand{\Pow}{\mathcal{P}}

\title{Possibilities determine the combinatorial structure of probability polytopes}

\author{Samson Abramsky, Rui Soares Barbosa, Kohei Kishida}
\address{Department of Computer Science, University of Oxford, Oxford, U.K.}
\email{\textup{\texttt{\{samson.abramsky\,|\,rui.soares.barbosa\,|\,kohei.kishida\}@cs.ox.ac.uk}}}

\author{Raymond Lal}
\address{Palantir Technologies, London, U.K.}
\email{rayashlal@gmail.com}

\author{Shane Mansfield}
\address{Institut de Recherche en Informatique Fondamentale \\ Universit\'e Paris Diderot -- Paris 7}
\email{shane.mansfield@univ-paris-diderot.fr}

\begin{document}

\begin{abstract}

We study the set of no-signalling empirical models on a  measurement scenario,
and show that the combinatorial structure of the no-signalling polytope
is completely determined by the possibilistic information given by the support of the models.
This is a special case of a general result which applies to all polytopes presented in a standard form, given by linear equations together with non-negativity constraints on the variables.
\end{abstract}

\maketitle

\section{Introduction}

This paper is concerned with the combinatorial structure of the polytopes of probability models which are widely studied in quantum information and foundations.
Much current study focusses on the \emph{no-signalling polytope} \cite{popescu2014nonlocality,pironio2011extremal}, comprising those probability models whose correlations are consistent with the constraints imposed by relativity, and on characterizing the set of quantum correlations \cite{NavascuesPironioAcin2008}, which is contained within this polytope.

As is well known, quantum correlations exceed those which can be achieved using classical, ``local realistic'' models. While Bell's original proof of this fact  \cite{Bell-thm} used the detailed probabilistic structure of the models arising from quantum mechanics to show that they violated the Bell inequalities which hold for classical models, subsequent proofs by Greenberger, Horne and Zeilinger \cite{GHZ,GHSZ90}, Hardy \cite{Hardy92:nonlocality1,Hardy93:nonlocality2}, and others \cite{cabello2001bell,zimba1993bell,Mermin90:QuantumMysteriesRevisited-SimplifiedGHZ1}, were inequality-free, and even probability-free. In \cite{AbramskyBrandenburger}, a hierarchy of forms of non-locality, or more generally contextuality, was defined. The higher levels, of \emph{logical contextuality} (generalizing Hardy arguments), \emph{strong contextuality} (generalizing GHZ arguments), and \emph{All-versus-Nothing contextuality} \cite{AbramskyEtAl:ContextualityCohomologyAndParadox}, use only the \emph{possibilistic} information from the probability models. That is, they need only the information about which events are possible (probability greater than zero) or impossible (probability zero). In other words, they only refer to the \emph{supports} of the probability distributions. Passing from probability models to their supports (the ``possibilistic collapse'') evidently loses a great deal of information.
In this paper, we show that nevertheless, the \emph{combinatorial structure} of the no-signalling polytope is completely captured by the possibilistically collapsed models, thus confirming that much structural information can in fact be gained from these apparently much simpler models.

In more precise terms, the \emph{combinatorial type} of a polytope is given by its \emph{face lattice} \cite{ziegler1995lectures}, the set of faces of the polytope, ordered by inclusion. These face lattices have a rich structure, and have been extensively studied in combinatorics.

Our main result can be stated as follows.
\begin{theorem}
Fix a ``measurement scenario'', specifying a set of variables which can be observed or measured, the possible outcomes of these measurements, and which variables are compatible and can be measured together. We can then define a polytope $\NS$ of no-signalling probability models over this scenario. Call the face lattice of this polytope $\FF$. Now let $\SP$ be the set of supports or possibilistic collapses of the models in $\NS$. $\SP$ is naturally ordered by context-wise inclusion of supports. Then there is an order-isomorphism
\[ \FF \cong \SP . \]
\end{theorem}
Thus the combinatorial type of the polytope $\NS$ is completely determined by its possibilistic collapse.

This result has a number of interesting corollaries. For example:
\begin{itemize}
\item All the models in the relative interior of each face $F \in \FF$ have the same support.
\item The vertices of $\NS$ are exactly the probability models with minimal support in $\NS$.
Moreover, there is only one probability model in $\NS$ for each such minimal support.
\item
The vertices of $\NS$ can be written as the disjoint union of the local, deterministic models --- the vertices of the polytope of classical models --- and the strongly contextual models with minimal support. 
\item Thus the extremal contextual probability models are completely determined by their supports.
\end{itemize}

In fact, this result  applies to a much wider class of polytopes. Note that the no-signalling polytope, for any given measurement scenario, is defined by the following types of constraint:
\begin{itemize}
\item Non-negativity
\item Linear equations: namely normalization, and the no-signalling conditions.
\end{itemize}
In geometric terms, this says that $\NS = \Hpos \cap \Aff (\NS)$, where $\Aff(\NS)$ is the affine subspace generated by $\NS$, and $\Hpos$ is the non-negative orthant, \ie the set of all vectors $\vv$ with $\vv \geq \zv$.

It is a standard result that every linear program can be put in a ``standard form'' of this kind \cite{matousek2007understanding}, so that its associated polytope of constraints is of the type we are considering.


Now our theorem in fact applies at this level of generality.

\begin{theorem}[General Version]
Let $P$ be a polytope such that $P = \Hpos \cap \Aff(P)$. Let $\FF(P)$ be the face lattice of $P$.
Let $\SP(P)$ be the set of ``supports'' of points in $P$, \ie $0/1$-vectors where each positive component of $\vv \in P$ is mapped to 1, while each zero component is left fixed.
$\SP(P)$ is naturally ordered componentwise, with $0 < 1$. Then there is an order-isomorphism
\[ \FF(P) \cong \SP(P) . \]
\end{theorem}

The structure of the remainder of the paper is as follows. In Section~2, we provide  background on measurement scenarios and empirical models. In Section~3 we review some basic notions on partial orders and lattices, and in Section~4, we review some standard material on polytopes. We prove our main result in Section~5, and apply it to probability polytopes, in particular no-signalling polytopes, in Section~6. 

\section{Measurement Scenarios and Empirical Models}
\label{empsec}

We shall give a brief introduction to the basic notions of measurement scenarios and empirical models, as developed in the sheaf-theoretic approach to contextuality and non-locality. For further discussion, motivation and technical details, see \cite{AbramskyBrandenburger,AbramskyEtAl:ContextualityCohomologyAndParadox}. An extended introduction and overview is given in \cite{Abr15}.

A \emph{measurement scenario} is a triple $(X, \MM, O)$ where:
\begin{itemize}
\item $X$ is a set of variables which can be measured, observed or evaluated

\item $\MM$ is a family of sets of variables, those which can be measured together. These form the \emph{contexts}.

\item $O$ is a set of possible outcomes or values for the variables.
\end{itemize}

In this paper, we shall only consider \emph{finite} measurement scenarios, where the sets $X$, $O$, and hence also $\MM$ are finite.
This allows us to avoid measure-theoretic technicalities, while capturing the primary objects of interest in quantum information and foundations.

Given a measurement scenario $(X, \MM, O)$, an \emph{empirical model} for this scenario is a family $\{ e_C \}_{C \in \MM}$ of probability distributions
\[ e_C \in \Prob(O^C), \quad C \in \MM . \]
Here we write $\Prob(X)$ for the set of probability distributions on a finite set $X$. Such distributions are simply represented by functions $d: X \rarr [0, 1]$ which are normalized:
\[ \sum_{x \in X} d(x) = 1 . \]
The probability of an event $S \subseteq X$ is then given by
\[ d(S) \; = \; \sum_{x \in S} d(x) . \]
The set $O^C$ is the set of all assignments $s : C \rarr O$ of outcomes to the variables in the context $C$.
Such an assignment represents a \emph{joint outcome} of measuring all the variables in the context. 

Given  $e_C \in \Prob(O^C)$, and a subset $U \subseteq C$, we have the operation of restricting $e_C$ to $U$ by \emph{marginalization}:
\[ e_C |_{U} \in \Prob(O^U) , \]
defined by:
\[ e_C |_{U}(s) \; = \; \sum_{t \in O^C, \; t |_U = s} e_C(t) . \]
Here $t |_U$ is the restriction of the assignment $t$ to the variables in $U$.

We say that an empirical model $\{ e_C \}_{C \in \MM}$ is \emph{no-signalling} if for all $C, C' \in \MM$:
\[ e_C |_{C \cap C'} \; = \; e_{C'} |_{C \cap C'} . \]
Thus an empirical model is no-signalling if the marginals from any pair of contexts to their overlap agree. This corresponds to a general form of an important physical principle.
Suppose that $C = \{ a, b \}$, and $C' = \{ a, b' \}$, where $a$ is a variable measured by an agent Alice, while $b$ and $b'$ are variables measured by Bob, who may be spacelike separated from Alice. Then under relativistic constraints, Bob's choice of measurement --- $b$ or $b'$ --- should not be able to affect the distribution Alice observes on the outcomes from her measurement of $a$. This is captured by saying that the distribution on $\{ a \} = \{ a, b \} \cap \{ a, b' \}$ is the same whether we marginalize from the  distribution  $e_C$, or the distribution $e_{C'}$.
The general form of this constraint is shown to be satisfied by the empirical models arising from quantum mechanics in \cite{AbramskyBrandenburger}.

\textbf{Example}
Consider the following table:
\begin{center}
\begin{tabular}{ll|ccccc}
A & B & $(0, 0)$ & $(1, 0)$ & $(0, 1)$ & $(1, 1)$  &  \\ \hline
$a$ & $b$ & $0$ & $1/2$ & $1/2$ & $0$ & \\
$a'$ & $b$ & $3/8$ & $1/8$ & $1/8$ & $3/8$ & \\
$a$ & $b'$ & $3/8$ & $1/8$ & $1/8$ & $3/8$ &  \\
$a'$ & $b'$ & $3/8$ & $1/8$ & $1/8$ & $3/8$ & 
\end{tabular}
\end{center}

This represents a situation where Alice and Bob can each choose measurement settings and observe outcomes. Alice can choose the settings $a$ or $a'$, while, independently,  Bob can choose $b$ or $b'$. The total set of variables is $X = \{ a, a', b, b' \}$.
The \emph{measurement contexts}  are
\[ \MM \; = \; \{\{ a, b \}, \;\; \{ a', b\}, \;\; \{ a, b' \}, \;\; \{ a', b' \} \}.\]
Each measurement has possible outcomes in $O = \{ 0, 1 \}$.
The matrix entry at row $(a', b)$ and column $(0,1)$ 
corresponds to the  \emph{event }
\[ \{ a' \mapsto 0,  \; b \mapsto 1 \} . \]
Each row of the table specifies a \emph{probability distribution} $e_C$ on events $O^C$ for a given choice of measurements $C$.
Thus the table directly corresponds to an empirical model on the measurement scenario $(X, \MM, O)$.

We verify that this table satisfies the no-signalling condition. Consider the following schematic representation of the table

\begin{center}
\begin{tabular}{ll|ccccc}
A & B & $(0, 0)$ & $(1, 0)$ & $(0, 1)$ & $(1, 1)$  &  \\ \hline
$a$ & $b$ & $c$ & $d$ & $e$ & $f$ & \\
$a'$ & $b$ & $g$ & $h$ & $i$ & $j$ & \\
$a$ & $b'$ & $k$ & $l$ & $m$ & $n$ &  \\
$a'$ & $b'$ & $o$ & $p$ & $q$ & $r$ & 
\end{tabular}
\end{center}
where we have labelled the entries with the letters $c$, \ldots , $r$. The no-signalling conditions for the non-empty intersections of contexts are given by the following equations:
\begin{center}
\begin{tabular}{lclclcl}
$c + e = k + m$, & $\quad$ & $d + f = l + n$, & $\quad$ & $g + i = o + q$, & $\quad$ & $h + j = p + r$ \\
$c + d = g + h$,  & $\quad$ & $e + f = i + j$,  & $\quad$ & $k + l = o + p$, & $\quad$ & $m + n = q + r$ \\
\end{tabular}
\end{center}
We see that, for example, the first equation is verified in the table above, since $0 + 1/2 = 3/8 + 1/8$.

This table has the additional property that it can be realized by a quantum state and appropriate choices of local observables.
It thus provides the basis for a proof of Bell's theorem \cite{Bell-thm}. See \cite{AbramskyBrandenburger} for an extended discussion.

Situations with any number of parties, measurement settings for each party, and outcomes can be represented similarly by empirical models on measurement scenarios.

\textbf{Example}
A different kind of example will illustrate the wide scope of our notion of empirical models over measurement scenarios.
Consider the measurement scenario $(X, \MM, O)$, where:
\begin{itemize}
\item $X$ is a set of 18 variables, $\{ A, \ldots , O \}$.
\item $\MM = \{ U_1 , \ldots , U_9 \}$, where the columns $U_i$ are the contexts:
\begin{center}
\begin{tabular}{|c|c|c|c|c|c|c|c|c|} \hline
$U_1$ & $U_2$ &  $U_3$ & $U_4$ & $U_5$ & $U_6$ & $U_7$ & $U_8$ & $U_9$ \\ \hline\hline
$A$ & $A$ & $H$ & $H$ & $B$ & $I$ & $P$ & $P$  & $Q$ \\ \hline
$B$ & $E$ & $I$ & $K$ & $E$ & $K$ & $Q$ & $R$ & $R$  \\ \hline
$C$ & $F$ & $C$ & $G$ & $M$ &  $N$ & $D$ & $F$ & $M$  \\ \hline
$D$ & $G$ & $J$ & $L$ & $N$  & $O$ & $J$ & $L$ & $O$  \\ \hline
\end{tabular}
\end{center}

\item $O = \{ 0, 1 \}$.
\end{itemize}
The empirical model has the property that the support of the distribution $e_C$ for each context $C$ consists of those assignments $s : C \rarr O$ which assign $1$ to exactly one variable in $C$.

This example shows that \emph{Kochen-Specker constructions} \cite{kochen1975problem},\cite{cabello1996bell} also fall within the scope of our definitions. See \cite{AbramskyBrandenburger} for further discussion.

\subsection*{Strong Contextuality}

The previous example also exhibits the property of \emph{strong contextuality}, which is a key part of the hierarchy of degrees of contextuality identified and studied in \cite{AbramskyBrandenburger}.
Let $e$ be an empirical model over a measurement scenario $(X, \MM, O)$.
We say that a global assignment $s : X \rarr O$  is consistent with the support of $e$, if  $s |_C \in \supp e_C$ for all $C \in \MM$. Here $\supp e_C := \{ s \in O^C \mid e_C(s) > 0 \}$.
We define $e$ to be \emph{strongly contextual} if there is no global assignment consistent with its support.

One can show that any empirical model satisfying the specification of the previous example is strongly contextual in this sense \cite{AbramskyBrandenburger}.

\subsection*{Vector representation of empirical models}
Each distribution $e_C$ can be specified by a list of $\card{O^C}$ numbers in $[0, 1]$, one for each assignment $s \in O^C$.
We can concatenate these lists to represent
the empirical model as a whole  by a vector in $\Real^n$, where
\[ n \; := \; \sum_{C \in \MM} \card{O^C} . \]
We shall henceforth pass freely from empirical models to their representation as real vectors without further comment.

\section{Interlude on partial orders and lattices}

We briefly pause to review some basic notions on partial orders and lattices. 
We refer to \cite{davey2002introduction} for additional background.

We recall that a  partial order is a set $P$ equipped with a binary relation $\leq$ which is reflexive, antisymmetric and transitive.
An element $\bot \in P$ is the least or bottom element if $\bot \leq x$ for all $x \in P$. Given two elements $x, y \in P$, their join or least upper bound, if it exists, is an element $z$ such that $x \leq z$, $y \leq z$, and whenever $x \leq w$ and $y \leq w$, then $z \leq w$. Note that such an element $z$ is necessarily unique if it exists; in this case, we write it as $x \vee y$.
If a poset $P$ has a least element, and a join $x \vee y$ for all elements $x, y \in P$, then we say that $P$ is a join semilattice.
Note that for any set $X$, the powerset $\Pow(X)$, ordered by inclusion, is a join semilattice, with the join being given by set union. The least element is the empty set.

The notion of meet or greatest lower bound is defined in the obvious dual fashion. Notation is $x \wedge y$ for the meet of elements $x$, $y$ if it exists. Also, we write $\top$ for the greatest or top element of a poset, if it exists. In the case of $\Pow(X)$, the meet is given by set intersection, and the top element is $X$. If a poset has a meet for every pair of elements, and a top element, it is a meet semilattice.

If a poset is both a join semilattice, and a meet semilattice, then it is a lattice.

Note that a join semilattice has a least upper bound for every finite subset of its elements. For subsets with more than two elements, this can be constructed by iterating the binary join operation.

Now suppose that we have a \emph{finite} join semilattice $P$. Then we can construct the meet of any subset $S \subseteq P$, by taking the join of the set $L$ of all the lower bounds of $S$ in $P$. Thus $P$ is in fact a lattice. A similar statement holds for a finite meet semilattice.

\section{Background on Polytopes}

We review some standard definitions and results on polytopes.
Our primary references for this section will be \cite{rockafellar2015convex} and \cite{ziegler1995lectures}.

We shall work concretely in $\Real^n$. We write $d(\xx,\yy)$ for the Euclidean distance function on $\Real^n$, and $\xx \cdot \yy$ for the inner product of vectors $\xx, \yy \in \Rn$. Vectors are ordered componentwise:
\[\xx \leq \yy \; \IFF \; (\forall i : 1 \leq i \leq n) \; \xx_i \leq \yy_i . \]
A \emph{closed half-space} in $\Rn$ is a set of the form
\[ \{ \xx \mid \av \cdot \xx \geq b \} \]
for some $\av \in \Rn$, $b \in \Real$. An \emph{$\HH$-polytope} is a bounded intersection of a finite set of closed half-spaces in $\Rn$. A \emph{$\VV$-polytope} is the convex hull $\Conv(S)$ of a finite set of points  $S \subseteq \Rn$. The Fundamental Theorem of polytopes \cite{ziegler1995lectures} says that these two notions coincide; we shall refer simply to polytopes.

An \emph{affine combination} of points $\xx_1, \ldots , \xx_k$ in $\Rn$ is an expression of the form
\[ \sum_{i=1}^{k} \lambda_i \xx_i , \qquad \sum_{i=1}^{k} \lambda_i = 1 . \]
We write $\Aff(S)$ for the set of all affine combinations of points of $S \subseteq \Rn$; this is the affine subspace generated by $S$.
Note that, if a linear equation $\av \cdot \xx = b$ is satisfied by all the elements of $S$, then it is satisfied by all the elements of $\Aff(S)$.

A linear inequality $\av \cdot \xx \geq b$ is \emph{valid} for a polytope $P$ if it is satisfied by every $\xx \in P$. Such a valid inequality defines a \emph{face} $F$ of $P$:
\[ F \; := \; \{ \xx \in P \mid \av \cdot \xx = b \} . \]
Any face is itself a polytope.
We write $\FF(P)$ for the set of faces of $P$, and $\Fp(P)$ for the set of non-empty faces.

The \emph{relative interior} of a polytope $P$ consists of those points $\xx \in P$ such that, for some $\epsilon > 0$, $\yy \in P$ whenever $\yy \in \Aff(P)$ and $d(\xx,\yy) \leq \epsilon$. We write $\relint P$ for the relative interior of $P$.

We will use the following characterization of the relative interior (\cite{rockafellar2015convex}[Theorem 6.4]).
\begin{theorem}
\label{relintthm}
A point $\xx \in P$ is in $\relint P$ if and only if for all $\yy \in P$, for some $\mu > 1$, $\mu \xx + (1-\mu) \yy$ is in $P$.
\end{theorem}
Intuitively, this says that a point $\xx$ of $P$ is in the relative interior if the line segment $[\yy, \xx]$ from any point $\yy$ of $P$ to $\xx$ can be extended beyond $\xx$ while remaining in $P$.

We will also use the following basic result (\cite{rockafellar2015convex}[Theorem 18.2]).
\begin{theorem}
Every polytope $P$ can be written as the disjoint union of the relative interiors of its non-empty faces:
\[ P \; = \; \bigsqcup_{F \in \Fp(P)} \relint F . \]
\end{theorem}
This means that for any  polytope $P$ we can define a map
\[ \carr : P \rTo \Fp(P) \]
which assigns to each point $\xx$ of $P$ its \emph{carrier face} --- the unique face $F$ such that $\xx \in \relint F$.

We can regard $\FF(P)$ as partially ordered by set inclusion. The following result is standard \cite{ziegler1995lectures}.
\begin{theorem}
$\FF(P)$ is a finite lattice. It is atomistic --- every element is the join of the atoms below it --- and coatomistic --- every element is the meet of the coatoms above it.\footnote{In the literature on polytopes, the terms ``atomic'' and ``coatomic'' are used, but these have a different meaning in the lattice theory literature \cite{davey2002introduction}.}
It is graded --- all maximal chains have the same length.
\end{theorem}
Note that meets in $\FF(P)$ are simply given by intersection of faces, while joins $F \vee G$ are defined indirectly, as the intersection of all faces containing both $F$ and $G$.

We call $\FF(P)$ the \emph{face lattice} of $P$. We refer to $\FF(P)$ as the \emph{combinatorial type} of $P$. Two polytopes with isomorphic face lattices are \emph{combinatorially equivalent}.

We say that a polytope $P$ is in \emph{standard form} if $P = \Hpos \cap \Aff(P)$. This means that $P$ is defined by a set of linear equations, together with the non-negativity constraint. It is a standard result of linear programming that every linear program can be put in this form \cite{matousek2007understanding}.
\begin{proposition}
\label{facesfprop}
If $P$ is in standard form, so is every face $F$ of $P$.
\end{proposition}
\begin{proof}
Let $F =  \{ \xx \in P \mid \av \cdot \xx = b \}$. Since $F \subseteq P$, $F \subseteq \Hpos \cap \Aff(F)$.
Conversely, since $\Aff(F) \subseteq \Aff(P)$, any $\xx \in \Hpos \cap \Aff(F)$ is in $P$, while $\xx \in \Aff(F)$ implies $\av \cdot \xx = b$. Hence $\xx \in F$, as required.
\end{proof}

We can define a map $\supp : \Hpos \rTo \{ 0, 1\}^n$:
\[ (\supp \xx)_i = \left\{ \begin{array}{ll}
0, & \xx_i = 0 \\
1, & \xx_i > 0
\end{array} \right.
\]
We call this the \emph{support} of a non-negative vector. In the case of a probability vector, this gives the usual notion of support of a distribution. In this case, we can also speak of \emph{possibilities}; an outcome is possible if it has positive probability, impossible if it has zero probability.

For a polytope $P$ in standard form, we define $\SL(P) := \{ \supp \xx \mid \xx \in P \}$. Since $\SL(P)$ is a subset of $\Rn$, it inherits the componentwise partial order on vectors. We write $\SL(P)_{\bot}$ for the result of adjoining a least element to this partially ordered set.
We make the following observation.

\begin{proposition}
Let $P$ be a polytope in standard form. If $\zv \in P$, then $P = \{ \zv \}$.
\end{proposition}
\begin{proof}
If $\zv \in P$, then the equations defining $\Aff(P)$ must all be homogeneous. If any non-zero $\vv$ is in $P$, the positive half-ray generated by $\vv$ will lie in $P$, contradicting the boundedness of $P$.
\end{proof}

Thus if $\card{P} > 1$, we can define $\SL(P)_{\bot} := \SL(P) \cup \{ \zv \}$ with the componentwise order.
We can define the join of two elements $\uu, \vv$ as the componentwise boolean disjunction (or equivalently, the supremum in $\Real$):
\[ (\uu \vee \vv)_i \; := \; \uu_i \vee \vv_i . \]
This is clearly the join (\ie pairwise supremum) in $\{ 0, 1 \}^n$.

We have the following simple result:
\begin{proposition}
\label{convsuppprop}
Let $P$ be a polytope in standard form, $\xx, \yy \in P$, $0 < \lambda < 1$. Then
\[ \supp (\lambda \xx + (1-\lambda) \yy) \; = \; \supp x \vee \supp y . \]
\end{proposition}

\begin{proposition}
$\SL(P)_{\bot}$ is a finite lattice.
\end{proposition}
\begin{proof}
The cases where $\card{P} \leq 1$ are trivial. Otherwise, the previous proposition shows that $\SL(P)$ is closed under the join operation.
This makes $\SL(P)_{\bot}$ a join semi-lattice, and since it is finite, it is a lattice.
\end{proof}

\section{Results}

We  fix a polytope $P$ in standard form. 
Given $\xx$ in $P$, we define a vector $\sx$ in $\Rn$:
\[ \sx_i = \left\{ \begin{array}{ll}
0, & \xx_i > 0 \\
1, & \xx_i = 0
\end{array} \right.
\]
Clearly $\sx \cdot \zz \geq 0$ is valid for all $\zz \in P$, and defines a face 
\begin{align*}
\Fx \, =\, & \setdef{\zz \in P}{\sx \cdot \zz = 0}\\
   = \,        & \setdef{\zz \in P}{\supp{\zz} \leq \supp{\xx}} .
\end{align*}

\begin{proposition}
\label{carrprop}
For all $\xx$ in $P$, $\carr \xx = \Fx$.
\end{proposition}
\begin{proof}
We show that $\xx \in \relint \Fx$, using Theorem~\ref{relintthm}.
Given $\zz \in \Fx$, for each $i$ with $\zz_i > 0$, for some $\epsilon_i > 0$, $\xx_i  \geq \epsilon_i \zz_i$.
Let $\epsilon = \min_i \epsilon_i$, $\mu = 1 + \epsilon$.
Then $\vv := \mu \xx + (1-\mu) \zz \geq \zv$, while since $\vv$ is an affine combination of points in $\Fx$, $\vv \in \Aff(\Fx)$. Thus by Proposition~\ref{facesfprop}, $\vv \in \Fx$, as required.
%
\end{proof}

The following is immediate from the definition of $\sx$.
\begin{proposition}
\label{orderequivprop}
For all $\xx, \yy \in P$: 
\[ \Fx \, \subseteq \, \Fy \IFF \supp \xx \leq \supp \yy . \]
\end{proposition}

We shall use a simple general result.

\begin{proposition}
\label{triangprop}
Let $X$ be a set, $Q$ and $R$ posets, and $f : X \repi Q$, $g : X \repi R$ surjective maps. Then the following are equivalent:
\begin{enumerate}
\item For all $x, y \in X$: $f(x) \leq f(y) \IFF g(x) \leq g(y)$.
\item There is a unique order-isomorphism $s : Q \rTo^{\cong}  R$ such that the following diagram commutes:
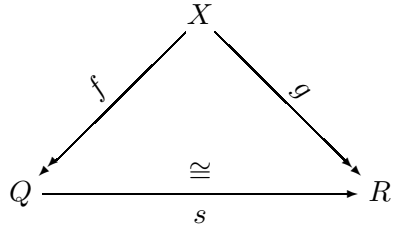
\begin{diagram}
& & X &  & \\
& \ldepi^{f} & & \rdepi^{g} & \\
Q &  & \rTo_{s}^{\cong}  & & R 
\end{diagram}
\end{enumerate}
\end{proposition}
\begin{proof}
Assume (1). Given $y \in Q$, since $f$ is surjective, $y = f(x)$ for some $x \in X$. Define $s(y) := g(x)$. The forward implication in (1) implies that $f(x) = f(y) \IMP g(x) = g(y)$, so this is well-defined. It also implies that $s$ is monotonic. The converse implication in (1) implies that $s$ is order-reflecting. Hence it is an order-isomorphism onto its image, which is $R$ by surjectivity of $g$. This isomorphism is unique by surjectivity of $f$.

The converse implication is immediate:
\[ f(x) \leq f(y) \IFF s(f(x)) \leq s(f(y)) \IFF g(x) \leq g(y) . \]
\end{proof}

We now come to our main result.

\begin{theorem}
There is an order-isomorphism $\FF(P) \cong \SL(P)_{\bot}$ between the face lattice and the support lattice of $P$, which sends a face to the support of any point in its relative interior.
\end{theorem}
\begin{proof}
We apply Propositions~\ref{triangprop}, \ref{carrprop} and \ref{orderequivprop}
 to the diagram
\begin{diagram}
& & P &  & \\
& \ldepi^{\carr} & & \rdepi^{\supp} & \\
\Fp(P) &  &   & & \SL(P)
\end{diagram}
This gives an order-isomorphism $\Fp(P) \cong \SL(P)$, which extends to a lattice isomorphism $\FF(P) \cong \SL(P)_{\bot}$ by mapping the empty face to the bottom element of $\SL(P)_{\bot}$.
\end{proof}

We now draw some corollaries of this result.

\begin{corollary}
\begin{enumerate}
\item Two points $\xx$, $\yy$ of $P$ have the same support if and only if they are in the relative interior of the same face.
\item The vertices of $P$ are exactly those points with minimal support.
\item A point $\xx$ of $P$ is a vertex if and only if it is the unique point of $P$ with support $\supp \xx$.
\end{enumerate}
\end{corollary}
\begin{proof}
(1) follows immediately from the definition of the order-isomorphism $s$. The atoms of $\FF(P)$ are the 0-dimensional faces $\{ \vv \}$ corresponding to the vertices of $P$; these are mapped bijectively by $s$ to the atoms of $\SL(P)$, which are the minimal supports, which yields (2).
Since the vertices of $P$ are exactly those points which are the unique elements of their carrier faces, (3) follows from (1).
\end{proof}

One interesting observation about our result is that we have shown an isomorphism between two lattices whose concrete presentations look rather different:
\begin{itemize}
\item For $\FF(P)$, meet is given simply by intersection, while join is defined indirectly, as the intersection of all upper bounds.
\item For $\SL(P)_{\bot}$, join is simply defined pointwise, while meet is defined indirectly, as the supremum of all lower bounds.
\end{itemize}

\section{Application to Probability Polytopes}

Let $\Sigma = (X, \MM, O)$ be a measurement scenario, and $\NS(\Sigma)$ the set of no-signalling models over $\Sg$.
As explained in Section~\ref{empsec}, we can view empirical models as vectors in $\Rn$, where $n$ depends on the measurement scenario.
Thus we can view $\NS(\Sigma)$ as a subset of $\Rn$. The points of $\NS(\Sg)$ are those vectors of non-negative numbers satisfying the normalization equations, one for each $C \in \MM$:
\[ \sum_{s \in O^C} e_C(s) = 1 \]
and the no-signalling equations, one for each pair $C, C' \in \MM$ and $s \in O^{C \cap C'}$:
\[ \sum_{t \in O^C, \; t |_{C \cap C'} = s} e_C(t) \; = \;  \sum_{t' \in O^{C'}, \; t' |_{C \cap C'} = s} e_C(t') . \]
We have written these in the notation of empirical models on measurement scenarios, but of course they can be ``flattened out'' into linear equations on the corresponding vectors.
Thus $\NS(\Sg)$ is specified by a set of linear equations, together with the non-negativity constraints. Hence $\NS(\Sg)$ is a polytope in standard form, and the results of the previous section apply.

We restate the results explicitly for $\NS(\Sg)$.

\begin{theorem}
\label{mainnsthm}
Let $\Sg$ be any measurement scenario.
There is an order-isomorphism 
\[ \FF(\NS(\Sg)) \cong \SL(\NS(\Sg))_{\bot} \]
between the face lattice and the support lattice of $\NS(\Sg)$.
\end{theorem}

The same result will hold for any polytope of probability models which can be put in standard form.

\begin{corollary}
\begin{enumerate}
\item Two empirical models in $\NS(\Sg)$ have the same support if and only if they are in the relative interior of the same face.
\item The vertices of $\NS(\Sg)$ are exactly those no-signalling empirical models with minimal support.
\item A no-signalling empirical model $e$ is a vertex if and only if it is the unique model in $\NS(\Sg)$ with support $\supp e$.
\item The empirical models in the relative interior of $\NS(\Sg)$ are exactly those with full support.
\end{enumerate}
\end{corollary}

We also have the following result specific to the no-signalling case.
\begin{proposition}
The vertices of $\NS(\Sg)$ can be written as a disjoint union
\[ V(\Sg) \; = \; \LD(\Sg) \,  \sqcup \, \MSC(\Sg) \]
of the local deterministic models $\LD(\Sg)$,  and $\MSC(\Sg)$, the strongly contextual models with minimal support.
The polytope $\LL(\Sg)$ of models realized by local hidden variables is given by  the convex hull $\Conv (\LD(\Sg))$, while every strongly contextual model over $\Sg$ is a convex combination of vertices in $\MSC(\Sg)$, and belongs to a face of $\NS(\Sg)$ containing only strongly contextual models.
\end{proposition}
\begin{proof}
By \cite{AbramskyBrandenburger}[Proposition 6.3], a model is strongly contextual if and only if it has no convex decomposition with a local model. Thus a strongly contextual model $e$ must be in the polytope $\Conv(\MSC(\Sg))$. Let $F = \carr e$. Any model $e' \in F$ has $\supp e' \leq \supp e$ by Theorem~\ref{mainnsthm}; hence the strong contextuality of $e$ implies that of $e'$.
\end{proof}

In these cases of probabilistic models, the supports have their usual interpretation as supports of probability distributions. Moreover, they acquire a conceptual significance as expressing \emph{possibilistic} information.

We can also view possibilistic models as well-motivated objects in their own right \cite{Abramsky:RelationalHiddenVariables,AbramskyBrandenburger,MansfieldFritz11:Hardy}. Boolean logic for combining possibilities replaces arithmetic formulas for calculating with probabilities.

In more precise terms, we use the  notion of \emph{commutative semiring}, which is an algebraic structure $(R, {+}, 0, {\cdot}, 1)$, where $(R, {+}, 0)$ and $(R, {\cdot}, 1)$ are commutative monoids satisfying the distributive law:
\[ a \cdot (b + c) = a\cdot b + a \cdot c . \]
Examples include the non-negative reals $\Rgeq$ with the usual addition and multiplication, and the booleans $\Bool = \{ 0, 1 \}$ with disjunction and conjunction playing the r\^oles of addition and multiplication respectively.

Now an $R$-distribution on a finite set $X$ is given by a map
\[ d : X \rTo R \]
satisfying the normalization condition
\[ \sum_{x \in X} d(x) \; = \; 1 . \]
In the case of $R = \Rgeq$, we recover the usual notion of probability distribution, while in the boolean case, a distribution is simply the characteristic function of a non-empty subset.

We can define marginalization of distributions over any commutative semiring $R$, exactly we did in the usual probabilistic case, and hence obtain a notion of no-signalling empirical model over a measurement scenario $\Sg$ with respect to $R$. We write $\NS(\Sg, R)$ for the set of all such models. In the case of $R = \Rgeq$, we recover the notion of probabilistic no-signalling model we have already seen, while in the boolean case, we obtain a notion of \emph{possibilistic no-signalling empirical model}.

We can now re-interpret the support map as a \emph{possibilistic collapse}, which maps a probabilistic model to a possibilistic one.

Some obvious questions arise:
\begin{itemize}
\item Are probabilistic no-signalling models mapped to possibilistic ones?
\item Does every possibilistic no-signalling model arise as the possibilistic collapse of a probabilistic no-signalling model?
\end{itemize}

These are answered by the results from \cite{Abramsky:RelationalHiddenVariables,AbramskyBrandenburger}, which show that the answer to the first question is positive, while the answer to the second is negative.
The first point follows from the fact that there is a (unique) semiring homomorphism $h : \Rgeq \rTo \{ 0, 1 \}$, which when applied pointwise to an empirical model yields the possibilistic collapse. The fact that $h$ is a homomorphism means that the no-signalling equations are preserved.
For the second point, we  give an explicit counter-example in \cite{Abramsky:RelationalHiddenVariables}[Proposition 9.1].

In the light of these results, we have the following situation. The support map can be viewed as the possibilistic collapse
\[ \poss : \NS(\Sg, \Rgeq) \rTo \NS(\Sg, \Bool) . \]
We know that this map is not surjective. 

We pose the following open question.
\begin{question}
Can we give an \emph{intrinsic characterization} of the image of the possibilistic collapse map, using only possibilistic notions?
\end{question}

%
%

It is tempting to attempt to answer to this question by conjecturing that the minimal models in $\NS(\Sg,\Bool)$ are in the image of the possibilistic collapse, and hence form the atoms of $\SL(\NS(\Sg))$. Since $\SL(\NS(\Sg))$ is atomistic, this would mean that the supports of the probabilistic models can all be expressed as joins of the minimal possibilistic models. This would provide a complete determination of the combinatorial structure of the no-signalling polytope by purely possibilistic information. The situation turns out to be somewhat more complicated, however, as we can demonstrate with the following possibilistic empirical model, which provides a counter-example to the tempting conjecture.

We shall now describe a simple example of a minimal possibilistic model $s$ which does not arise as the support of any probabilistic model; i.e.~such that there exists no probabilistic model $e$ for which $\mathsf{poss}(e) = s$. Note that this also implies that there cannot be any probabilistic model with strictly smaller support, since otherwise its possibilistic collapse would contradict the minimality of $s$.

The measurement scenario is described as follows.
\begin{align*}
X &= \{A,B,C,D\} \\
\mathcal{M} &= \left\{ \{A,B\}, \{A,C\}, \{A,D\}, \{B,C\}, \{B,D\}, \{C,D\} \right\} \\
O &= \{0,1,2\}
\end{align*}
The possibilistic model $s$ is then defined by the following possible sections whose coefficients we label $a,b,\dots,o$.
\begin{equation*}
\begin{array}{lcccc}
AB & \mapsto & 00, & 10, & 21 \\
& & a & b & c \\
AC & \mapsto & 00, & 11, & 21 \\
& & d & e & f \\
AD & \mapsto & 01, & 10, & 21 \\
& & k & l & m \\
BC & \mapsto & 00, & 11 & \\
& & g & h & \\
BD & \mapsto & 00, & 11 & \\
& & i & j & \\
CD & \mapsto & 01, & 10 & \\
& & n & o & \\
\end{array}
\end{equation*}
The model is also depicted in the bundle diagram of figure \ref{fig:model}.
What we see at the bottom in this representation is the \emph{base space} of the variables in $X$. There is an edge between two variables when they can be measured together.
Above each vertex  is a \emph{fibre} of those values which can be assigned to the variable. We represent the values in the order $0$, $1$, $2$. There is an edge between values in adjacent fibres precisely when the corresponding \emph{joint outcome} is possible. For example, in the context $AD$ the joint outcome $01$ is possible, so there is an edge connecting these values in the respective fibres above $A$ and $D$; while $00$ is not possible, so there is no corresponding edge.
For further examples of how empirical models are represented as bundle diagrams, see \cite{AbramskyEtAl:ContextualityCohomologyAndParadox}.

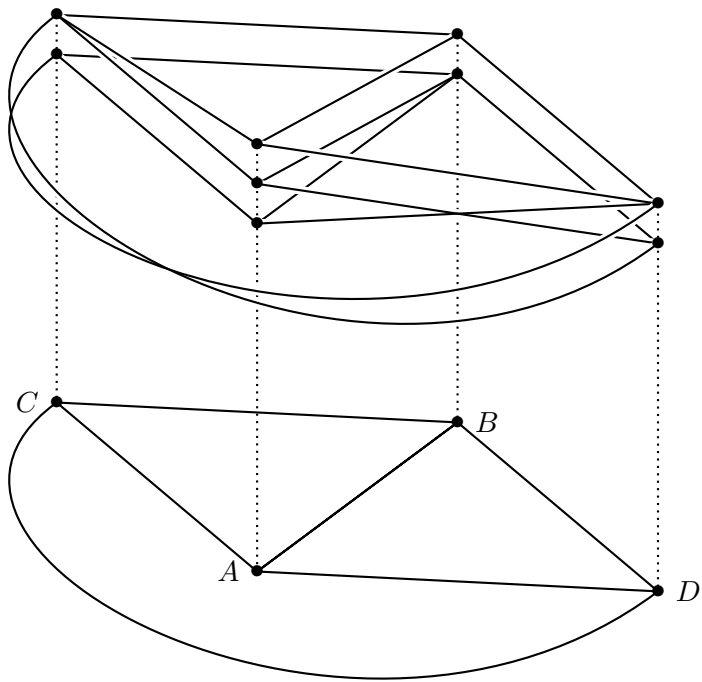
\begin{figure}
\caption{\label{fig:model} The possibilistic model $s$.}
\begin{center}
\begin{tikzpicture}[x=50pt,y=50pt,thick,label distance=-0.25em,baseline=(current bounding box.center),scale=.75]

\node (e) at (4,-0.2) {};
\node (n) at (2,1.5) {};
\node (T) at (0,3.5) {};
\node (u) at (0,0.4) {};

\node [inner sep=0em] (a) at (0,0) {};
\node [inner sep=0em] (b) at ($ (a) + (n) $) {};
\node [inner sep=0em] (c) at ($ (a) + (n) - (e) $) {};
\node [inner sep=0em] (d) at ($ (a) + (e) $) {};
\node [inner sep=0em] (c') at ($ (a) - (e) $) {};
\node [inner sep=0em] (d') at ($ (a) + (e) - 1.5*(n) $) {};

\node [inner sep=0em] (a0) at ($ (a) + (T) $) {};
\node [inner sep=0em] (a1) at ($ (a0) + (u) $) {};
\node [inner sep=0em] (a2) at ($ (a0) + 2*(u) $) {};
\node [inner sep=0em] (b0) at ($ (b) + (T) $) {};
\node [inner sep=0em] (b1) at ($ (b0) + (u) $) {};
\node [inner sep=0em] (c0) at ($ (c) + (T) $) {};
\node [inner sep=0em] (c1) at ($ (c0) + (u) $) {};
\node [inner sep=0em] (d0) at ($ (d) + (T) $) {};
\node [inner sep=0em] (d1) at ($ (d0) + (u) $) {};
\node [inner sep=0em] (c'0) at ($ (c') + (T) $) {};
\node [inner sep=0em] (c'1) at ($ (c'0) + (u) $) {};
\node [inner sep=0em] (d'0) at ($ (d') + (T) $) {};
\node [inner sep=0em] (d'1) at ($ (d'0) + (u) $) {};

\draw (a) -- (b) -- (c) -- (a) -- (d) -- (b) -- (a);
\draw (c) .. controls (c') and (d') .. (d);

\draw [dotted] (a2) -- (a);
\draw [dotted] (b1) -- (b);
\draw [dotted] (c1) -- (c);
\draw [dotted] (d1) -- (d);

\node [inner sep=0.2em,label=left:{$A$}] at (a) {$\bullet$};
\node [inner sep=0.2em,label=right:{$B$}] at (b) {$\bullet$};
\node [inner sep=0.2em,label=left:{$C$}] at (c) {$\bullet$};
\node [inner sep=0.2em,label=right:{$D$}] at (d) {$\bullet$};

\draw (b0) -- (c0);
\draw (b1) -- (c1);

\draw [line width=3.2pt,white] (c0) -- (a0);
\draw [line width=3.2pt,white] (c1) -- (a1);
\draw [line width=3.2pt,white] (c1) -- (a2);

\draw (c0) -- (a0);
\draw (c1) -- (a1);
\draw (c1) -- (a2);

\draw (b0) -- (d0);
\draw (b1) -- (d1);

\draw [line width=3.2pt,white] (a0) -- (b0);
\draw [line width=3.2pt,white] (a1) -- (b0);
\draw [line width=3.2pt,white] (a2) -- (b1);

\draw (a0) -- (b0);
\draw (a1) -- (b0);
\draw (a2) -- (b1);

\draw [line width=3.2pt,white] (a0) -- (d1);
\draw [line width=3.2pt,white] (a1) -- (d0);
\draw [line width=3.2pt,white] (a2) -- (d1);

\draw (a0) -- (d1);
\draw (a1) -- (d0);
\draw (a2) -- (d1);

\draw [line width=3.2pt,white] (c0) .. controls (c'0) and (d'1) .. (d1);
\draw [line width=3.2pt,white] (c1) .. controls (c'1) and (d'0) .. (d0);

\draw (c0) .. controls (c'0) and (d'1) .. (d1);
\draw (c1) .. controls (c'1) and (d'0) .. (d0);

\node [inner sep=0.2em] at (a0) {$\bullet$};
\node [inner sep=0.2em] at (a1) {$\bullet$};
\node [inner sep=0.2em] at (a2) {$\bullet$};
\node [inner sep=0.2em] at (b0) {$\bullet$};
\node [inner sep=0.2em] at (b1) {$\bullet$};
\node [inner sep=0.2em] at (c0) {$\bullet$};
\node [inner sep=0.2em] at (c1) {$\bullet$};
\node [inner sep=0.2em] at (d0) {$\bullet$};
\node [inner sep=0.2em] at (d1) {$\bullet$};

\end{tikzpicture}
\end{center}
\end{figure}

We first observe that no-signalling requires all coefficients to be equated.
For example, the no-signalling conditions for the overlapping contexts $AB$, $AC$ translate into the equations
\[ a = d, \quad b = e, \quad c = f. \]
Continuing in this fashion, we obtain the following equations:
\[ \begin{array}{llllllll}
a = k, & b = l, &  g = i, & h = j, & c = n, & d = k, & e = l, & f = m \\
c = h, & h = o, & g = n, & i = o, & j = n, & c = j, & l = o, & d = n.
\end{array}
\]
These equations imply that all the coefficients must be equated.
Minimality of the model can be deduced from the fact that if the coefficient for any of the possible sections is set to zero then consistency requires that all coefficients must be set to zero.

Having equated all coefficients, the only remaining non-trivial equation is
\begin{equation*}
a = a + a.
\end{equation*}
Since this arises from a possibilistic model, it evidently has a non-zero solution in the booleans, namely $a = 1$. Observe, however, that there exists no non-zero, real solution. This implies that there exists no probabilistic model $e$ such that $\mathsf{supp}(e) = s$.

Further, we note that such possibilistic models can also arise in  ``Bell-type'' measurement scenarios of the kind described in the first example in Section~2. In \cite{MansfieldBarbosa:QPL2013} a method was introduced for constructing Bell-type models from models of more general kinds such that the degree of contextuality is preserved. Under this construction, the model $s$ described above yields a model $s_{\text{Bell}}$ with two agents, each of which have four measurement settings, with each measurement having three possible outcomes, which also has the property that there exists no probabilistic model $e$ such that $\mathsf{supp}(e) = s_{\text{Bell}}$.

The new measurement scenario is described as follows.

\begin{align*}
X_{\text{Bell}} &= \{A_1,B_1,C_1,D_1,A_2,B_2,C_2,D_2\} \\
\mathcal{M}_{\text{Bell}} &= \{A_1,B_1,C_1,D_1\} \times \{A_2,B_2,C_2,D_2\} \\
O &= \{0,1,2\}
\end{align*}

The possibilistic model $s_{\text{Bell}}$ has possible sections defined as follows (c.f.~figure \ref{fig:bellmodel}).
\begin{equation*}
\begin{array}{lllclll}
A_1A_2 & & & \mapsto & 00, & 11, & 22 \\
B_1B_2, & C_1C_2, & D_1D_2 & \mapsto & 00, & 11 & \\
A_1B_2, & A_2B_1 & & \mapsto & 00, & 10, & 21 \\
A_1C_2, & A_2C_1 & & \mapsto & 00, & 11, & 21 \\
A_1D_2, & A_2D_1 & & \mapsto & 01, & 10, & 21 \\
B_1C_2, & B_2C_1 & & \mapsto & 00, & 11 & \\
B_1D_2, & B_2D_1 & & \mapsto & 00, & 11 & \\
C_1D_2, & C_2D_1 & & \mapsto & 01, & 10 & 
\end{array}
\end{equation*}

\begin{figure}
\caption{\label{fig:bellmodel} The possibilistic model $s_\text{Bell}$.}
\begin{center}
\begin{tikzpicture}[x=50pt,y=50pt,thick,label distance=-0.25em,baseline=(current bounding box.center),scale=.45]

\node (e) at (4.5,-0.2) {};
\node (n) at (2,4) {};
\node (T) at (0,6) {};
\node (u) at (0,0.5) {};

\node [inner sep=0em] (a) at (0,0) {};
\node [inner sep=0em] (b) at ($ (a) + (e) $) {};
\node [inner sep=0em] (c) at ($ (a) + 2*(e) $) {};
\node [inner sep=0em] (d) at ($ (a) + 3*(e) $) {};
\node [inner sep=0em] (a') at ($ (a) + (n) $) {};
\node [inner sep=0em] (b') at ($ (b) + (n) $) {};
\node [inner sep=0em] (c') at ($ (c) + (n) $) {};
\node [inner sep=0em] (d') at ($ (d) + (n) $) {};

\node [inner sep=0em] (a0) at ($ (a) + (T) $) {};
\node [inner sep=0em] (a1) at ($ (a0) + (u) $) {};
\node [inner sep=0em] (a2) at ($ (a0) + 2*(u) $) {};
\node [inner sep=0em] (a'0) at ($ (a0) + (n) $) {};
\node [inner sep=0em] (a'1) at ($ (a1) + (n) $) {};
\node [inner sep=0em] (a'2) at ($ (a2) + (n) $) {};
\node [inner sep=0em] (b0) at ($ (b) + (T) $) {};
\node [inner sep=0em] (b1) at ($ (b0) + (u) $) {};
\node [inner sep=0em] (b'0) at ($ (b0) + (n) $) {};
\node [inner sep=0em] (b'1) at ($ (b1) + (n) $) {};
\node [inner sep=0em] (c0) at ($ (c) + (T) $) {};
\node [inner sep=0em] (c1) at ($ (c0) + (u) $) {};
\node [inner sep=0em] (c'0) at ($ (c0) + (n) $) {};
\node [inner sep=0em] (c'1) at ($ (c1) + (n) $) {};
\node [inner sep=0em] (d0) at ($ (d) + (T) $) {};
\node [inner sep=0em] (d1) at ($ (d0) + (u) $) {};
\node [inner sep=0em] (d'0) at ($ (d0) + (n) $) {};
\node [inner sep=0em] (d'1) at ($ (d1) + (n) $) {};

\draw (a) -- (a');
\draw (a) -- (b');
\draw (a) -- (c');
\draw (a) -- (d');
\draw (b) -- (a');
\draw (b) -- (b');
\draw (b) -- (c');
\draw (b) -- (d');
\draw (c) -- (a');
\draw (c) -- (b');
\draw (c) -- (c');
\draw (c) -- (d');
\draw (d) -- (a');
\draw (d) -- (b');
\draw (d) -- (c');
\draw (d) -- (d');

\draw [dotted] (a2) -- (a);
\draw [dotted] (b1) -- (b);
\draw [dotted] (c1) -- (c);
\draw [dotted] (d1) -- (d);
\draw [dotted] (a'2) -- (a');
\draw [dotted] (b'1) -- (b');
\draw [dotted] (c'1) -- (c');
\draw [dotted] (d'1) -- (d');

\node [inner sep=0.2em,label=below:{$A_1$}] at (a) {$\bullet$};
\node [inner sep=0.2em,label=below:{$B_1$}] at (b) {$\bullet$};
\node [inner sep=0.2em,label=below:{$C_1$}] at (c) {$\bullet$};
\node [inner sep=0.2em,label=below:{$D_1$}] at (d) {$\bullet$};
\node [inner sep=0.2em,label=left:{$A_2$}] at (a') {$\bullet$};
\node [inner sep=0.2em,label=left:{$B_2$}] at (b') {$\bullet$};
\node [inner sep=0.2em,label=right:{$C_2$}] at (c') {$\bullet$};
\node [inner sep=0.2em,label=right:{$D_2$}] at (d') {$\bullet$};

\draw [thin] (a0) -- (c'0);
\draw [thin] (b0) -- (d'0);

\draw [thin] (a0) -- (b'0);
\draw [thin] (b0) -- (c'0);

\draw [thin] (a0) -- (a'0);
\draw [thin] (b0) -- (b'0);
\draw [thin] (c0) -- (c'0);
\draw [thin] (d0) -- (d'0);

\draw [thin] (b0) -- (a'0);
\draw [thin] (c0) -- (b'0);

\draw [thin] (c0) -- (a'0);
\draw [thin] (d0) -- (b'0);

\draw [thin] (a0) -- (d'1);
\draw [thin] (d1) -- (a'0);

\draw [thin] (a1) -- (b'0);

\draw [thin] (d1) -- (c'0);

\draw [thin] (c0) -- (d'1);
\draw [thin] (b0) -- (a'1);

\draw [thin] (a1) -- (d'0);
\draw [thin] (c1) -- (d'0);

\draw [thin] (a1) -- (c'1);
\draw [thin] (b1) -- (d'1);
\draw [thin] (b1) -- (c'1);
\draw [thin] (a1) -- (a'1);
\draw [thin] (b1) -- (b'1);
\draw [thin] (c1) -- (c'1);
\draw [thin] (d1) -- (d'1);
\draw [thin] (c1) -- (a'1);
\draw [thin] (d1) -- (b'1);
\draw [thin] (c1) -- (b'1);

\draw [thin] (a2) -- (a'2);

\draw [thin] (a2) -- (b'1);

\draw [thin] (a2) -- (c'1);

\draw [thin] (a2) -- (d'1);

\draw [thin] (b1) -- (a'2);

\draw [thin] (c1) -- (a'2);

\draw [thin] (c1) -- (d'0);

\draw [thin] (d1) -- (a'2);

\draw [thin] (d0) -- (c'1);

\node [inner sep=0.2em] at (a0) {$\bullet$};
\node [inner sep=0.2em] at (a1) {$\bullet$};
\node [inner sep=0.2em] at (a2) {$\bullet$};
\node [inner sep=0.2em] at (b0) {$\bullet$};
\node [inner sep=0.2em] at (b1) {$\bullet$};
\node [inner sep=0.2em] at (c0) {$\bullet$};
\node [inner sep=0.2em] at (c1) {$\bullet$};
\node [inner sep=0.2em] at (d0) {$\bullet$};
\node [inner sep=0.2em] at (d1) {$\bullet$};
\node [inner sep=0.2em] at (a'0) {$\bullet$};
\node [inner sep=0.2em] at (a'1) {$\bullet$};
\node [inner sep=0.2em] at (a'2) {$\bullet$};
\node [inner sep=0.2em] at (b'0) {$\bullet$};
\node [inner sep=0.2em] at (b'1) {$\bullet$};
\node [inner sep=0.2em] at (c'0) {$\bullet$};
\node [inner sep=0.2em] at (c'1) {$\bullet$};
\node [inner sep=0.2em] at (d'0) {$\bullet$};
\node [inner sep=0.2em] at (d'1) {$\bullet$};

\end{tikzpicture}
\end{center}
\end{figure}
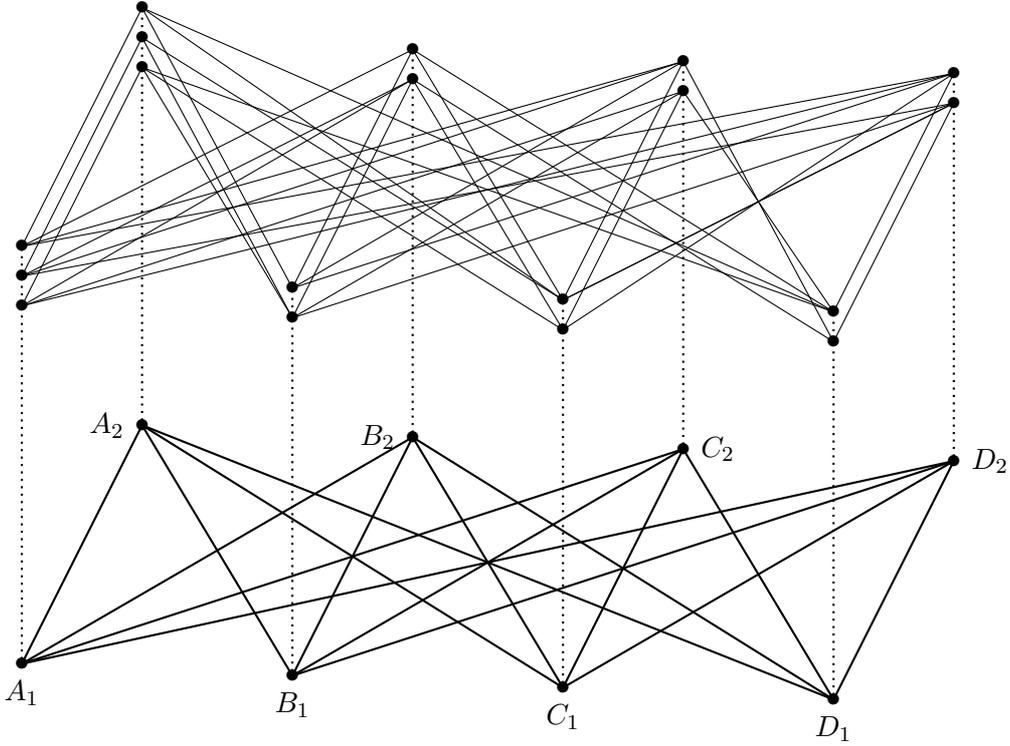

Note that maximal contexts in $s$ each contain two elements, so that to model these in an $n$-partite scenario requires $n=2$. The sets $X_{\text{Bell}}$ and $\mathcal{M}_{\text{Bell}}$ together with $O$ define the smallest bipartite measurement scenario which models the contexts $\mathcal{M}$ in bipartite form. For all $P,Q \in X$, the possible sections of $s_\text{Bell}$ at the context $P_1Q_2$ are defined either by the possible sections of $s$ at the context $PQ$ in the case that $P \neq Q$, or by the diagonal sections $\{oo \mid o \in \mathsf{supp}(s_{\{P\}})\}$ in the case that $P=Q$. The equivalence in terms of contextuality of the original and constructed Bell models was established in \cite{MansfieldBarbosa:QPL2013} and covers the type of contextuality (strong, logical, probabilistic) as well as the degree to which the model violates the analogous contextual inequalities. For our present purposes, it is also clear that the system of consistency equations yielded by a constructed Bell model is equivalent to the system of consistency equations for the original model, once the simple identifications of coefficients arising from the diagonalised contexts are taken into account.


\section*{Acknowledgements}
Support from the following is gratefully acknowledged: EPSRC EP/K015478/1, Quantum Mathematics and Computation (SA); Fondation Sciences Math\'{e}matiques de Paris, postdoctoral research grant eotpFIELD15RPOMT-FSMP1, Contextual Semantics for Quantum Theory (SM); FCT -- Funda\c{c}\~ao para a Ci\^encia e Tecnologia
(Portuguese Foundation for Science and Technology), 
PhD grant SFRH/BD/94945/2013 (RSB); John Templeton Foundation, Categorical Unification (RSB); U.S. AFOSR FA9550-12-1-0136, Topological and Game-Semantic Methods for Understanding Cyber-Security (SA, KK); the Oxford Martin School James Martin Program on Bio-inspired Quantum Technologies, (SA, SM, RSB); and Templeton World Charity Foundation (RL).

We are grateful to Janne Kujala for a number of helpful comments which led to improvements and clarifications in the presentation.
We also thank the anonymous journal referee for their comments.
\bibliographystyle{apalike}
\bibliography{nspoly}

\end{document}